\documentclass[10pt,a4paper]{article}

\usepackage{amsmath,amsgen,latexsym}
\usepackage{amstext,amssymb,amsfonts,latexsym}
\usepackage{theorem}
\usepackage{pifont}
\usepackage[dvips]{graphics,epsfig}

 \newcommand{\bs}{\bigskip}
 \newcommand{\ms}{\medskip}
 \newcommand{\n}{\noindent}
 
 \newcommand{\hs}[1]{\hspace*{ #1 mm}}
 \newcommand{\vs}[1]{\vspace*{ #1 mm}}

 \newcommand{\setempty}{\mathrm{\O}}
 
 \newcommand{\nat}{\mathbb{N}}

 \newcommand{\co}{\mathrm{co}\mbox{-}}

 \newcommand{\ie}{\textrm{i.e.},\hspace*{2mm}}
 \newcommand{\eg}{\textrm{e.g.},\hspace*{2mm}}

 \newcommand{\AAA}{{\cal A}}
 \newcommand{\BB}{{\cal B}}
 \newcommand{\CC}{{\cal C}}
 \newcommand{\FF}{{\cal F}}
 \newcommand{\DD}{{\cal D}}
 \newcommand{\EE}{{\cal E}}

 \newcommand{\GG}{{\cal G}}

 \newcommand{\PP}{{\cal P}}

 \newcommand{\p}{\mathrm{P}}

 \newcommand{\reg}{\mathrm{REG}}
 \newcommand{\cfl}{\mathrm{CFL}}

 \newcommand{\onecequallin}{1\mbox{-}\mathrm{C}_{=}\mathrm{LIN}}
 \newcommand{\oneplin}{1\mbox{-}\mathrm{PLIN}}

\theoremstyle{plain}
\theoremheaderfont{\bfseries}
 \newtheorem{theorem}{Theorem}[section]
 \newtheorem{lemma}[theorem]{Lemma}
 \newtheorem{proposition}[theorem]{Proposition}
 \newtheorem{corollary}[theorem]{Corollary}

 {\theorembodyfont{\rmfamily}
  }
 {\theorembodyfont{\rmfamily} \newtheorem{example}[theorem]{Example}}
 {\theorembodyfont{\rmfamily} }

 \newtheorem{claim}{Claim}

 \newenvironment{proof}{\par \noindent
            {\bf Proof. \hs{2}}}{\hfill$\Box$ \vspace*{3mm}}

 \newenvironment{proofof}[1]{\vspace*{5mm} \par \noindent
         {\bf Proof of #1.\hs{2}}}{\hfill$\Box$ \vspace*{3mm}}

 \newenvironment{yproof}{\par \noindent
            {\bf Proof. \hs{2}}}{\hfill$\Box$ \vspace*{3mm}}

\newcommand{\ignore}[1]{}

\newcommand{\track}[2]{[\:\begin{subarray}{c} #1 \\%
      #2 \end{subarray} ]}

\newcommand{\regdissect}{\mathrm{REG}\mbox{-}\mathrm{DISSECT}}
\newcommand{\semilin}{\mathrm{SEMILIN}}
\newcommand{\cgl}{\mathrm{CGL}}
\newcommand{\bcfl}{\mathrm{BCFL}}
\newcommand{\bcflbh}{\mathrm{BCFL}_{\mathrm{BH}}}

\begin{document}
%%%%%%%%%%%%%%%%%%
%%%%%%%%%%%%%%%%%%
%%%%%%%%%%%%%%%%%%
%%%
\pagestyle{plain}
\setcounter{page}{1}

\begin{center}
{\Large {\bf The Dissecting Power of Regular Languages}} \bs\ms\\

{\sc Tomoyuki Yamakami\footnote{Present Affiliation: Department of Information Science, University of Fukui, 3-9-1 Bunkyo, Fukui 910-8507, Japan.} \hs{1} and \hs{1}  Yuichi Kato}$^{*}$ \bs\\
\end{center}

%%%%%%%%%%%%%%%%%
%%%%%%%%%%%%%%%%%

\begin{quote}
\n{\bf Abstract.} 
A recent study on structural properties of regular and context-free languages has greatly    promoted our basic understandings of the complex behaviors of those languages. We continue the study to examine how regular languages behave  when they need to cut numerous infinite languages. A particular interest rests on a situation in which a regular language needs to ``dissect''  a given infinite language into two subsets of infinite size.
Every context-free language is dissected by carefully  chosen regular languages (or  it is REG-dissectible). In a larger picture, we show that  constantly-growing languages and semi-linear languages are REG-dissectible.
Under certain natural conditions, complements and finite intersections of semi-linear languages also become REG-dissectible.
Restricted to bounded languages, the intersections of finitely many context-free languages and, more surprisingly, the entire Boolean hierarchy over bounded context-free languages are REG-dissectible.
As an immediate application of the REG-dissectibility, we show another structural property, in which an appropriate bounded context-free language can ``separate with infinite margins''  two given nested infinite bounded context-free languages.

\ms
\n{\bf keywords.} 
theory of computing,  formal languages,  regular language,  context-free language, 
bounded language,  semi-linear,  constantly growing,  dissectible,  i-separate
\end{quote}

%%%%%%%%%%%%%%%%%%%%%%%%%%%%%%%%%%
%%%%%%%%%%%%%%%%%%%%%%%%%%%%%%%%%%
\section{Background Knowledge and the Results' Overview}\label{sec:introduction}

The exquisitely complex behaviors of formal languages are often dictated by multiple-layers of inner structures of the languages and a mathematical theory over those languages has been developed in the past six decades alongside the discovery of some of the hidden structures.
In an early stage of the study of context-free languages, for instance,  a notion of
{\em semi-linearity}---a structural property on the frequency of occurrences of symbols---was found  in  \cite{Par61} and  a {\em pumping lemma}---another property regarding the growth rate of strings---was proven in \cite{BPS61}.
Similarly,  underlying structures of regular languages have been analyzed within a number of different frameworks, including the Myhill-Nerode theorem, monadic second-order logic, and finitely generated monoids.
Recently, new realms of structural properties of languages have been studied by obvious analogy with structural complexity issues of polynomial time-bounded complexity classes.  Such properties include {\em primeimmunity} as well as {\em pseudorandomness} against the regular and context-free languages, introduced in \cite{Yam11}, and a notion of {\em minimal cover}, which was applied to the regular languages in \cite{DSY02}. In the literature, numerous key questions concerning the behaviors of languages have been raised but left unsolved. We suspect that the difficulty in answering those questions may be rooted in yet-unknown structures that constitute the languages.

To promote our understandings of formal languages in general, it may be desirable to unearth  the hidden structural properties of the languages. In this line of study, this paper aims at exploring another structural property, which is seemingly innocent but possibly fundamental,  concerning the ability to partition a target infinite set into two portions of infinite size.
This simple property, which we name ``dissectibility,'' seems  more suitable for weak  computations, because, as shown in Section \ref{sec:how-to-dissect}, polynomial-time decidable languages, for instance, are powerful enough to dissect any recursive languages of infinite size. Among models of weak computations, we are focused on the regular languages, because they are generally regarded as weak in recognition power;
however, they could exhibit surprisingly high  power in dissecting infinite languages.
To be more precise at this point, an infinite set $C$ is said to {\em dissect} a target infinite set $L$, as illustrated in Fig.\ref{fig:dissection}, if two disjoint sets $C\cap L$ and $\overline{C}\cap L$ ($=L-C$) are both infinite, where $\overline{C}$ expresses  the {\em complement} of $C$. When $C$ is particularly a regular language, we succinctly say that  $L$ is {\em REG-dissectible}.
We are mostly interested in clarifying exactly what kind of languages are $\reg$-dissectible.
A typical example of $\reg$-dissectible language is the aforementioned context-free languages (Corollary \ref{CFL-REG-dissect}).
As for another example, let us consider a language $L_1$ generated by a grammar whose productions include a special form $S\rightarrow SS$, where $S$ is the start symbol. Irrelevant to its computational complexity, the language $L_1$ can be dissected by a regular language composed of strings of lengths that are equal to zero modulo $3$, because
 $L_1$ contains a series of strings of lengths $2k,3k,4k,\ldots$ for an appropriately chosen constant $k>0$. A more concrete example is  the language  $L_2=\{w^{n!}\mid w\in\{a,b\}^2,n\in\nat\}$.
Although this language $L_2$ is not even context-free, it can be easily dissected by a regular language consisting of  strings, each of which begins with the letter $a$. The third example   language is  $L_3 = \{(ab^n)^n \mid n\in\nat\}$, whose complement is context-free. This language $L_3$ can be easily dissected by a regular language whose strings contain an even number of $a$'s.
As a relevant notion, a $\CC$-pseudorandom language \cite{Yam11} also dissects any language in $\CC$ with quite large {\em margins}, where the intuitive term ``margin'' refers to the difference between two given sets.

%%%
%%%
\begin{figure}[b]
\begin{minipage}{0.5\hsize}
 \begin{center}
 \includegraphics[width=3.5cm]{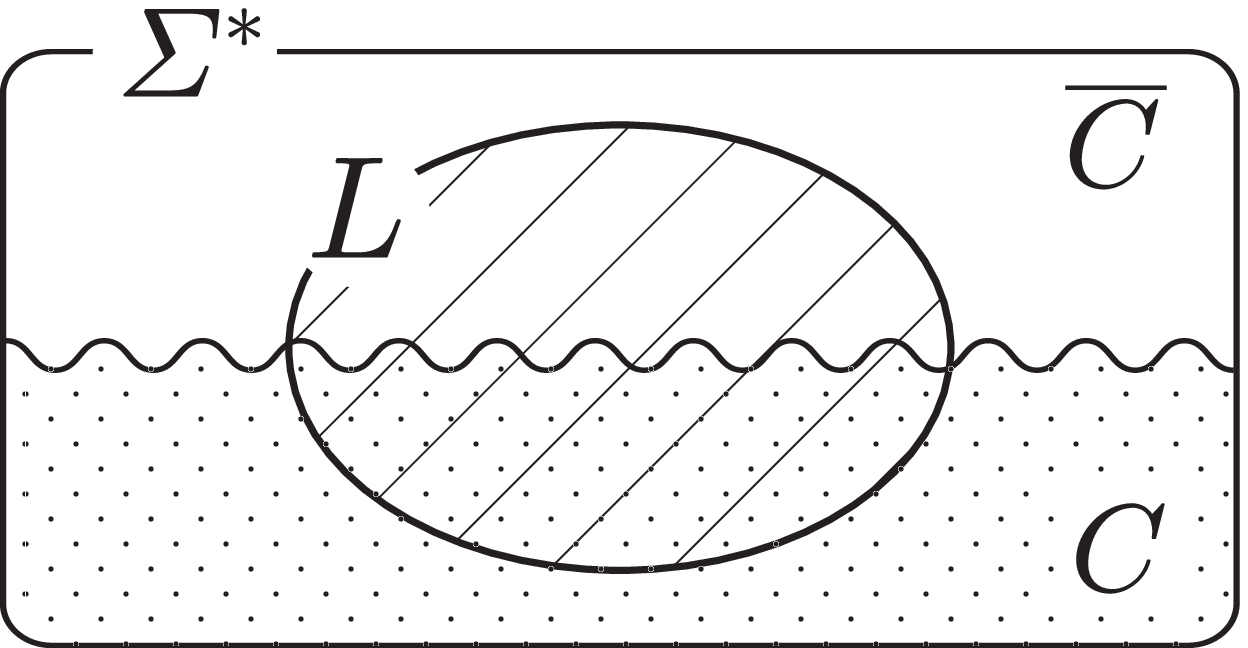}
 \end{center}
\caption{$C$ dissects $L$.}\label{fig:dissection}
\end{minipage}
\begin{minipage}{0.5\hsize}
 \begin{center}
 \includegraphics[width=3.5cm]{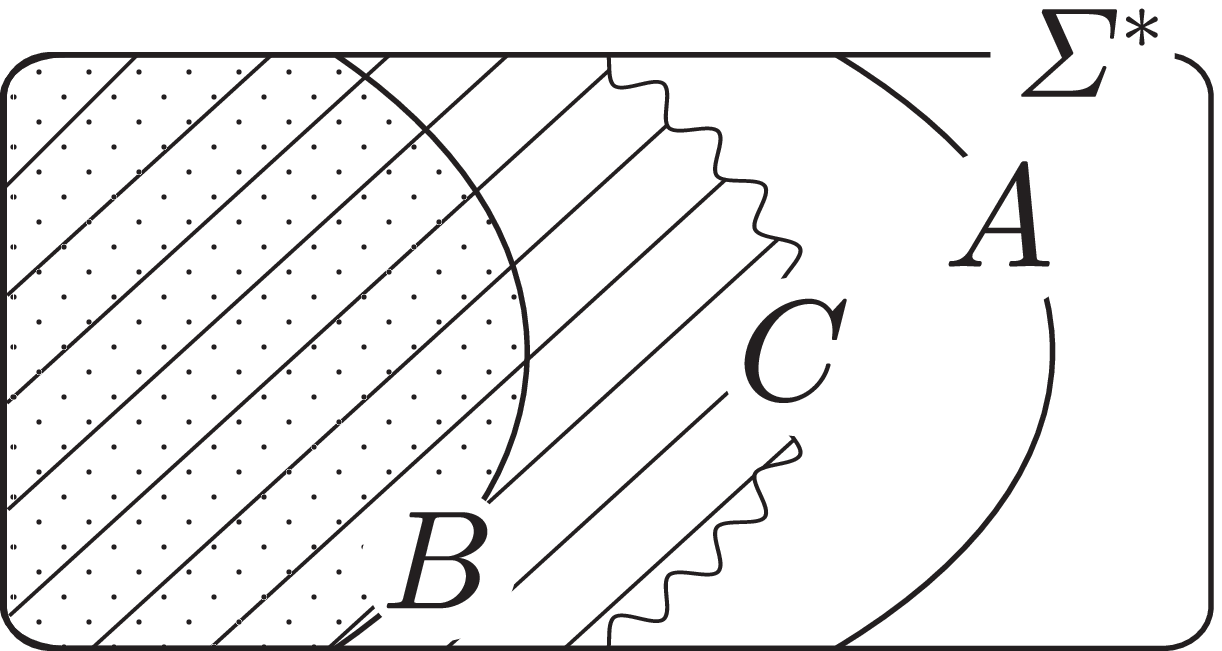}
 \end{center}
\caption{$C$ i-separates $\mathrm{i}(B,A)$.}\label{fig:separation}
\end{minipage}
\end{figure}
%%%
%%%

Through Sections \ref{sec:how-to-dissect} to \ref{sec:context-free}, two wider families of languages, {\em constantly-growing languages} and {\em semi-linear languages}, will be  shown to be $\reg$-dissectible.  Under certain natural conditions, the complements, the intersections, and the differences of semi-linear languages are proven to be $\reg$-dissectible using a
simple analysis of length patterns of strings inside a given language.
This analysis involves a manipulation of solutions of semi-linear equations and those conditions are indeed necessary to guarantee the $\reg$-dissectibility.
On the contrary, a rather obvious limitation exists for the $\reg$-dissectibility; namely, as  shown in Section \ref{sec:how-to-dissect},
there is a logarithmic-space computable language that cannot be $\reg$-dissectible (Theorem \ref{impossible-L}).
Taking a step further forward, when limited to {\em bounded languages} of Ginsburg and Spanier \cite{GS64}, we will be able to show that
the intersections of finitely many context-free languages are dissected by appropriate regular languages, despite the fact that the intersections of $k$ bounded context-free languages for $k\geq1$ form an infinite hierarchy within the family of context-sensitive languages \cite{LW73}.
By elaborating our argument further, we will prove that the entire {\em Boolean hierarchy} over the class of bounded context-free languages is also $\reg$-dissectible (Theorem \ref{Boolen-dissect}). These results will be presented in Section \ref{sec:context-free}.

The $\reg$-dissectibility notion has intimate connections to other notions.
Earlier, Domaratzki, Shallit, and Yu \cite{DSY02} studied a notion of minimal cover, which means the ``smallest'' superset $A$ of a given set $B$, where  ``smallest'' means that there is no set between $A$ and $B$ with infinite margins.
Motivated by their notion and results, we pay a special attention to a structural property of separating two infinite ``nested'' languages with infinite margins.
In our term of ``separation with infinite margins'' (or {\em i-separation}, in short), we actually mean, as illustrated in Fig.\ref{fig:separation}, that a pair of infinite sets $A$ and $B$, denoted by $\mathrm{i}(B,A)$, for which $A$ covers $B$ with an infinite margin,
can be separated by an appropriate set $C$ that lies in between the two sets with infinite margins.
As an immediate application of the aforementioned $\reg$-dissectibility results for the  bounded context-free languages, we will show in Section \ref{sec:application}
that two bounded context-free languages can be i-separated by bounded context-free languages.
This i-separation result will be further extended into any level of the Boolean hierarchy over bounded context-free languages (Theorem \ref{BCFL-separate}).

{}From the next section, we will formally introduce the key notions of the $\reg$-dissectibility and the i-separation and we will present detailed proofs of our major results mentioned above.

%%%%%%%%%%%%%%%%%%%%%%%%%%%%%%%%%%%%
%%%%%%%%%%%%%%%%%%%%%%%%%%%%%%%%%%%%
\section{Notions and Notations}

We briefly explain a set of basic notions and notations that will be used in the subsequent sections.
First, we denote by $\nat$ the set of all {\em natural numbers} (\ie nonnegative integers) and we write  $\nat^{+}$  for  $\nat-\{0\}$.  For each number $n\in\nat^{+}$, the notation $[n]$ denotes the {\em integer interval}  $\{1,2,3,\ldots,n\}$.
Associated with three arbitrary numbers $a,b,k\in\nat$, we define  $A_{a,b,k}$ to be the set $\{an+b\mid n\in\nat, n\geq k\}$. The generic notation $O$ denotes both an all-zero vector  and  an all-zero matrix of appropriate dimension.
For two sets $A$ and $B$, the set $\{x \mid x\in A, x\not\in B\}$ is the {\em difference} between $A$ and $B$ and is expressed as $A-B$.
When $A$ is a {\em countable} set, the succinct notation $|A|=\infty$ (resp., $|A|<\infty$) indicates that $A$ is an infinite (resp., a finite) set. Given two countable sets $A$ and $B$, we write $A\subseteq_{ae}B$ to mean $|A-B|<\infty$, and  the  notation $A=_{ae}B$ is used whenever  both $A\subseteq_{ae}B$ and $B\subseteq_{ae}A$ hold, where the subscript ``ae'' stands for ``almost everywhere.''

An {\em alphabet} $\Sigma$ is a finite nonempty set of ``symbols'' and a {\em string} over $\Sigma$ is a finite sequence of symbols in $\Sigma$.  The set of all strings over $\Sigma$ is denoted $\Sigma^*$, and $\Sigma^+$ expresses the set $\Sigma^*-\{\lambda\}$, where $\lambda$ is the {\em empty string}. The {\em length} $|x|$ of any string $x$ is the total number of occurrences of symbols in $x$. For any string $x$ and any symbol $\sigma$, the notation $\#_{\sigma}(x)$ stands for
the number of occurrences of $\sigma$ in $x$.
Given a language $S$, the {\em length set} of $S$, denoted $LT(S)$, is the collection of all lengths $|x|$ for any strings $x$ in $S$. We often identify a language $S$ with its {\em characteristic function}, which is also denoted $S$ (\ie $S(x)=1$ if $x\in S$, and $S(x)=0$ otherwise).
The sets of all regular languages and of all context-free languages are expressed respectively as  $\reg$ and $\cfl$.

The {\em complement} of a language $B$ over alphabet $\Sigma$ is the set $\Sigma^*-B$ and it is denoted $\overline{B}$ as far as its underlying alphabet $\Sigma$ is clear from the context.
For ease of our notations, we use the following four class operations:  (1) $\CC\wedge \DD =\{C\cap D\mid C\in\CC,D\in\DD\}$, (2) $\CC\vee \DD =\{C\cup D\mid C\in\CC,D\in\DD\}$, (3)  $\CC-\DD =\{C-D\mid C\in\CC,D\in\DD\}$, and (4) $\co\CC = \{\overline{C}\mid C\in\CC\}$, where $\CC$ and $\DD$ are language families.
Given any family $\FF$ of languages, a language $S$ is said to be {\em $\FF$-immune} if $S$ is infinite and $S$ has no infinite subset belonging to $\FF$ (see, \eg \cite{Yam11}).

%%%%%%%%%%
%%%%%%%%%%
\section{How to Dissect Languages}\label{sec:how-to-dissect}

Let us recall from Section \ref{sec:introduction} that an infinite language $S$ is {\em REG-dissectible} exactly when there exists a regular language $C$ that dissects $S$ (\ie $|C\cap S|= |\overline{C}\cap S|=\infty$). Moreover, a nonempty language family $\FF$ is  {\em REG-dissectible} if and only if every infinite language in $\FF$ is $\reg$-dissectible. Notice that, since  this definition disregards
all {\em finite} languages inside $\FF$,  we implicitly assume that $\FF$ contains infinite languages. We can naturally expand the $\reg$-dissectibility to a more general notion of $\CC$-dissectibility simply by replacing $\reg$ with an arbitrary nonempty language family $\CC$; however,  the choice of $\reg$ is actually of great importance.
In fact,  it is more interesting to consider low-complexity language families like $\reg$ as a candidate for $\CC$. One reason is that polynomial-time decidable  languages, for instance,  are powerful enough to dissect any infinite recursive  languages.

\begin{example}
We claim that every infinite recursive language is $\p$-dissectible, where $\p$ is the family of all polynomial-time decidable languages.
Let $L$ be any infinite language over alphabet $\Sigma$ recognized by a two-way single-tape deterministic Turing machine $M$ that eventually halts on all inputs. For simplicity, let $\Sigma=\{0,1\}$ and assume that $L\neq_{ae}\Sigma^*$ because, otherwise, a regular set $C=\{0x \mid x\in\Sigma^*\}$ easily dissects $L$.
Now, we define $C$ as follows. Let $z_0,z_1,z_2,\ldots$ be a standard lexicographic order of all strings over $\Sigma$. Given each string $x$,  to determine the value $C(x)$, we  go through
the following procedure $\PP$ from round $0$ to round $|x|$. Initially, we set $A=R=\setempty$. At round $i$, we  first  compute the value $C(z_i)$ by calling $\PP$ recursively round by round.
We then simulate $M$ on the input $z_i$ within $|x|$ steps. When  $M(z_i)=1$, we update $A$ to  $A\cup\{i\}$ if $C(z_i)=1$, and $R$ to $R\cup\{i\}$ if $C(z_i)=0$.  On the contrary, when  either $M(z_i)=0$ or $M(z_i)$ is not obtained within $|x|$ steps, we
do nothing. After round $|x|$, if $|A|>|R|$, then define $C(x) = 0$;  otherwise, define $C(x)=1$. Clearly, $C$ is in $\p$. By a diagonalization argument, we can show that $|C\cap L|=|\overline{C}\cap L|=\infty$.  Therefore, every infinite recursive language can be dissected by an appropriate language in $\p$.
\end{example}

In the following second example, we will show that a simple use of {\em advice} makes it possible to dissect arbitrary languages by appropriate regular languages. For basic properties of the advice, the reader may refer to \cite{TYL10,Yam10,Yam11}.

\begin{example}
We claim that every infinite language is $\reg/n$-dissectible, where $\reg/n$ is the collection of {\em advised regular languages}, each of which is of the form $\{x \mid M \text{ accepts }\track{x}{h(|x|)}\}$ for an appropriate {\em deterministic finite automaton} (or dfa), an advice alphabet $\Gamma$, and an advice function $h:\nat\rightarrow\Gamma^*$ satisfying $|h(n)|=n$ for all $n\in\nat$, where $\track{x}{y}$ is a {\em track notation} used in  \cite{TYL10}.
To verify this claim, take any infinite language $L$ over alphabet $\Sigma$. Since $L$ is infinite, the length set $LT(L)$ is also infinite. Hence, we partition $LT(L)$ into two infinite subsets, say, $S_1$ and $S_2$; that is, $S_1\cap S_2=\setempty$,  $LT(L) = S_1\cup S_2$, and $|S_1|=|S_2|=\infty$. Without loss of generality, we assume that $0\not\in S_1$. Now, let us define an advice function $h:\nat\rightarrow\{0,1\}^*$ as  $h(n) = 10^{n-1}$ if $n\in S_1$ and $h(n)=0^{n}$ otherwise. We also define a dfa $M$ that behaves as follows: on input $\track{x}{y}$, if $y= 10^{|x|-1}$ with $|x|\geq1$, then $M$ accepts the input; otherwise, it rejects the input. The language $C=\{x\mid \text{$M$ accepts $\track{x}{h(|x|)}$}\}$ then belongs to $\reg/n$. Obviously, for any string $x\in L$ with $|x|\in S_1$, since $h(|x|)=10^{|x|-1}$, $M$ accepts $\track{x}{h(|x|)}$.  It  thus holds that $|C\cap L| = \infty$. Similarly, for any $x\in S$ with $|x|\in S_2$, $M$ rejects $\track{x}{h(|x|)}$, implying $|\overline{C}\cap L|=\infty$. In conclusion, $C$ dissects $L$.
\end{example}

As noted in Section \ref{sec:introduction}, a pattern of the lengths of strings in a target language surely plays a key role in  proving its  $\reg$-dissectibility. This fact turns our attention to languages composed of strings satisfying a certain length condition, known as a ``constant growth property.''  Formally, a nonempty language $L$ is said to be {\em constantly growing} if there exist a constant $p>0$ and a finite subset $K\subseteq\nat^{+}$ that meet the following condition: for every string $x$ in $L$ with $|x|\geq p$, there exist  a string $y\in L$ and a constant $c\in K$ for which $|x|=|y|+c$ holds. Such languages can be easily dissected by appropriately chosen regular languages as shown in the next lemma.

\begin{lemma}\label{const-growth-dissect}
Every infinite constantly-growing language is $\reg$-dissectible.
\end{lemma}

\begin{yproof}
Let $L$ be any infinite language over alphabet $\Sigma$ and assume that $L$ is constantly growing with a constant $p>0$ and a finite set  $K\subseteq\nat^{+}$. Now, let $c$ denote the maximal element in $K$ and set $c'=c+1$. For each index $i\in[c]$, we take a special language $L_i=\{x\in L\mid |x|\equiv i\;(\mathrm{mod}\;c')\}$, and we wish to prove that at least two distinct indices $i_1,i_2\in[c]$ satisfy that $|L_{i_1}|=|L_{i_2}|=\infty$. Toward a contradiction, we assume otherwise.   Since $L = \bigcup_{i\in[c]}L_i$,
exactly one index $i\in[c]$ must make $L_i$ infinite.  Let us fix such an index, say, $i$. Given any index $j\in[c]$, we  set  $S_{i,j}$ to be $\{y\in L\mid \exists x\in L_i\,[\,|x|=|y|+j\,]\}$. Since $L$ is  constantly growing, a set $S_{i,j}$ must be infinite for a certain index $j$.  Note that $S_{i,j}\subseteq L_{\ell}$ holds for $\ell = i-j\;\mathrm{mod}\;c'$.
This containment implies that $L_{\ell}$ is infinite, contradicting the uniqueness of $i$  since $i\neq\ell$. Therefore, we can choose two distinct indices $i_1,i_2\in[c]$ for which  $|L_{i_1}|=|L_{i_2}|=\infty$.
Finally, we define $C=\{x\in\Sigma^*\mid |x|\equiv i_1\;(\mathrm{mod}\;c')\}$, which is clearly  regular.  Since $L_{i_1}\subseteq C$ and $L_{i_2}\subseteq \overline{C}$, it obviously follows that
$|C\cap L|=|\overline{C}\cap L|=\infty$. In other words, $C$ dissects $L$, as requested.
\end{yproof}

For a wider application of Lemma \ref{const-growth-dissect}, it is desirable to strengthen the lemma slightly. In what follows, we succinctly write $\cgl$ for the family of all constantly-growing languages and use the notion of  {\em CGL-immunity} to describe our proposition.

\begin{proposition}\label{CG-immune-dissect}
Every language that is not $\cgl$-immune is $\reg$-dissectible.
\end{proposition}

The above proposition comes from Lemma \ref{const-growth-dissect} as well as the following {\em transitive closure property} of $\reg$-dissectibility:
for any two infinite languages $A$ and $B$, if $A$ is $\reg$-dissectible and $A\subseteq B$, then $B$ is also $\reg$-dissectible.

Luckily, a length pattern of strings in  a language is not the only feature used to dissect the target language. For example,  the languages $L_2$ and $L_3$ exemplified in Section \ref{sec:introduction} are not constantly growing; however, they are  dissected by regular languages. Before presenting more examples of $\reg$-dissectible languages in the next section,
we will show a plausible limitation of the dissecting power of the regular languages.
Following a standard convention, the notation $\mathrm{L}$ stands for the family of all languages that can be recognized by two-way deterministic Turing machines using a read-only input tape together with a constant number of logarithmic space-bounded read/write work tapes.
In the next proposition, we will show that $\mathrm{L}$ contains a language that cannot be  dissected by any regular languages.

\begin{theorem}\label{impossible-L}
The language family $\mathrm{L}$ is not $\reg$-dissectible.
\end{theorem}

\begin{yproof}
Let us consider the unary language $S=\{0^{n!}\mid n\in\nat\}$ over the alphabet $\Sigma=\{0\}$.
Firstly, we will show that $S$ is in $\mathrm{L}$. For this purpose, it suffices to design a logarithmic-space deterministic Turing machine that recognizes $S$. On input of the form $0^m$, the desired machine $M$ writes $m$ in binary on its 1st work tape using $O(\log{m})$ cells and $1$ on its 2nd work tape. At each round, $M$ reads out a number, say, $n$ in binary written on the 2nd tape and checks if $m$ is a multiple of $n$ using the 3rd work tape as a counter up to $n$. If not, then $M$  immediately rejects the input; otherwise, it  increases $n$ by one (in binary) before entering the next round. If the machine does not reject until $n$ reaches $m$, then it accepts the input.

Secondly,  we want to show that no regular language can dissect $S$. Assume otherwise; that is, there exists an infinite language $C\in\reg$  over $\Sigma$ that dissects $S$.
We need the following technical property (Claim \ref{REG-character}) of this unary regular language $C$ regarding its length set $LT(C)$. Let us recall the notation $A_{a,b,k}$ and, in addition, set $\GG=\{(a,b,k)\mid a,b,k\in\nat, b<a\}$ for the description of the property.

\begin{claim}\label{REG-character}
For any unary language $C$, $C$ is regular iff there exists a finite set $G\subseteq \GG$ for which $LT(C) = \bigcup_{(a,b,k)\in G} A_{a,b,k}$.
\end{claim}

Claim \ref{REG-character} is attributed to Parikh \cite{Par61} and, since $C\in\reg$, the claim  guarantees the existence of a finite set $G$  that characterizes $C$; namely, $LT(C) = \bigcup_{(a,b,k)\in G}A_{a,b,k}$.

Since $|C\cap S|=\infty$, there exists a triplet $(a,b,k)$ in $G$ satisfying $|\{m\mid \exists\, n\geq k\,[\,m! = an+b\,]\}|=\infty$. Now, we argue that $b=0$.  First, take two integers $m,n$ with $n\geq k$ and $m>a$ satisfying $an+b = m!$. Since $a<m$,  $m!\equiv 0\;(\mathrm{mod}\;a)$ holds.  {}From  $an+b\equiv b\;(\mathrm{mod}\;a)$, we obtain $b\equiv 0\;(\mathrm{mod}\;a)$. Since $b<a$, $b$ must be zero, as requested.
Moreover, it holds that  $a>1$. To see this fact, suppose that $a=1$. Since $A_{1,0,k}$ equals
$\{n \mid n\geq k\}$,  we conclude that  $|\nat - A_{1,0,k}| <\infty$. Therefore, it follows that $|LT(\overline{C})\cap LT(S)|<\infty$, contradicting  $|\overline{C}\cap S|=\infty$.

Since $a>1$ and $b=0$, for a certain large constant $k'$, it holds that
$\{m!\mid m\geq k'\}\subseteq A_{a,0,k}$. This implies that $|LT(\overline{C})\cap LT(S)|<\infty$. This is a clear contradiction, and therefore $C$ cannot dissect $S$.
\end{yproof}

For convenience, we denote by $\regdissect$ the collection of all {\em infinite}  $\reg$-dissectible languages. It is not difficult to prove the following closure/non-closure properties.
(1) The set $\regdissect$ is closed under concatenation, reversal, Kleene star, and union. (2)
$\regdissect$ is not closed under intersection with regular languages.
(3) Moreover, $\regdissect$ is not closed under $\lambda$-free homomorphism as well as under quotient with regular languages, where $\lambda$ is the empty string. The last two properties can be proven using certain languages derived from the one presented in the proof of Theorem  \ref{impossible-L}.

%%%%%%%%%%%%%%%%%%%%%%%%%
%%%%%%%%%%%%%%%%%%%%%%%%%
\section{Context-Free Languages and Bounded Languages}\label{sec:context-free}

Parikh \cite{Par61} discovered that the number of times that each symbol  occurs  in each string of a given context-free language $L$ must satisfy a certain system of linear  Diophantine equations.
This result inspired a notion of {\em semi-linear languages}.
Context-free languages are an important example of semi-linear languages and a semi-linear nature of languages will be exploited in certain cases of the $\reg$-dissectibility proofs of the languages.  First, we will explain the notion of semi-linear sets and languages using a {\em matrix formalism}.
A subset $A$ of $\nat^{k}$ is called {\em linear} if
there exist a number $m\in\nat$ and an $(m+1)\times k$ nonnegative integer matrix (called a {\em critical matrix}) $T$ satisfying the following condition: for every point $v\in\nat^{k}$, $v$ is in $A$ if and only if  $(1,z_1,z_2,\ldots,z_m)T = v$
holds for a certain tuple (called a {\em solution})   $(z_1,z_2,\ldots,z_m)\in\nat^{m}$.
A {\em semi-linear set} is a union of finitely many linear sets.
Given any string $x$ over alphabet $\Sigma=\{\sigma_1,\sigma_2,\ldots,\sigma_k\}$,  a {\em Parikh image} of $x$,  denoted by  $\Psi(x)$, is  a point $(\#_{\sigma_1}(x),\#_{\sigma_2}(x),\ldots,\#_{\sigma_k}(x))$ in the space $\nat^{k}$,
and the  {\em commutative image} (or the {\em Parikh image})  $\Psi(L)$ of a language $L$ over $\Sigma$ refers to the set $\{\Psi(x)\mid x\in L\}$.
A language $L$ is called {\em semi-linear} whenever $\Psi(L)$ is semi-linear.

The family of all semi-linear languages is denoted by $\semilin$, and $\semilin(2)$ expresses the family   $\semilin\wedge \semilin$.

\begin{lemma}\label{SEMILIN-dissect}
$\semilin \subseteq \regdissect$ but $\semilin(2)\nsubseteq \regdissect$.
\end{lemma}

\begin{proof}
Every semi-linear language $L$ is defined by a finite set of certain linear equations  and this fact proves that $L$ has the property of constant growth.
Lemma \ref{const-growth-dissect} therefore leads to the first part of the lemma.
To see that $\semilin(2)$ is not $\reg$-dissectible, let us consider two  example languages    $L_1=\{0^n1^n\mid n\in\nat\}$ and $L_2=\{1^n0^n\mid n\in\nat\}\cup \{0^{n!}1^{n!}\mid n\in\nat\}$ over the binary alphabet $\Sigma=\{0,1\}$. Since  $\Psi(L_1)=\Psi(L_2)=\{(n,n)\mid n\in\nat\}$, $L_1$ and $L_2$ are semi-linear. However, the intersection $L_1\cap L_2\in\semilin(2)$, which equals $\{0^{n!}1^{n!}\mid n\in\nat\}$,  can be shown to be non-$\reg$-dissectible by an argument similar to the proof of Theorem  \ref{impossible-L}.
\end{proof}

Since $\cfl\subseteq\semilin$ \cite{Par61}, Lemma \ref{SEMILIN-dissect} immediately yields the following consequence.

\begin{corollary}\label{CFL-REG-dissect}
The language family $\cfl$ is $\reg$-dissectible.
\end{corollary}

To utilize well-studied properties on semi-linear languages, we limit our attention within
a restricted part of context-free languages. A language $L$ over alphabet $\Sigma$ is said to be {\em bounded} if there are fixed nonempty
strings $w_1,w_2,\ldots,w_m$ in $\Sigma^*$ such that $L$ is a subset of  $L[w_1,w_2,\ldots,w_m]=_{\mathrm{def}}\{w_1^{i_1}w_2^{i_2}\cdots w_m^{i_m}\mid i_1,i_2,\ldots,i_m\in\nat\}$ \cite{GS64}.
For readability, we abbreviate as $\bcfl$ the family of all bounded context-free languages.
The {\em $k$-conjunctive closure} of $\bcfl$, denoted $\bcfl(k)$,  is
defined inductively as follows: $\bcfl(1)=\bcfl$ and $\bcfl(k)=\bcfl(k-1)\wedge \bcfl$ for every index $k\geq2$.
Earlier, Liu and Weiner \cite{LW73} proved that the collection $\{\bcfl(k)\mid k\in\nat^{+}\}$ forms an infinite hierarchy within the family  of context-sensitive languages.

\begin{theorem}\label{BCFL(k)-dissect}
For any index $k\geq1$, $\bcfl(k)$ is $\reg$-dissectible.
\end{theorem}

For the proof of Theorem \ref{BCFL(k)-dissect}, we define $\tilde{\Psi}(w)$ to be  $\{(i_1,i_2,\ldots,i_m)\in\nat^{m}\mid w=w_1^{i_1}w_2^{i_2}\cdots w_m^{i_m}\}$ for each string $w$ in $L[w_1,w_2,\ldots,w_m]$.
Notice that $\tilde{\Psi}(w)$ could contain numerous elements
because $w$ may have more than one expression of the form $w^{i_1}_1w^{i_2}_2\cdots w^{i_m}_m$.
Finally, we define $\tilde{\Psi}(L)= \bigcup_{w\in L}\tilde{\Psi}(w)$ for any bounded language $L$.  This operator $\tilde{\Psi}$ works similarly as $\Psi$ does and, by exploiting this similarity, Ginsburg \cite{Gin66} exhibited a close relationship between a bounded context-free language $L$ and the semi-linearity of  $\tilde{\Psi}(L)$. What we need for our proof given below is the following slightly weaker form of \cite[Theorem 5.4.2]{Gin66}:
for any subset $L$ of $L[w_1,\ldots,w_k]$ in $\bcfl$,  $\tilde{\Psi}(L)$ is semi-linear, and thus $L$  belongs to $\semilin$.

\begin{proofof}{Theorem \ref{BCFL(k)-dissect}}
We start with the following general claim regarding $\Psi$.
By viewing $w_1,w_2,\ldots,w_{m}$ as ``different'' symbols $\sigma_1,\sigma_2,\ldots,\sigma_{m}$ as in \cite{Gin66},  a similarity between $\Psi(w)$ and $\tilde{\Psi}(w)$ makes  the claim  true for $\tilde{\Psi}$ as well.

\begin{claim}\label{semilin-two-intersection}
For any languages $L_1,L_2\in\semilin$, if $|L_1\cap L_2|=\infty$ and $\Psi(L_1)\cap \Psi(L_{2})\subseteq \Psi(L_1\cap L_2)$ hold, then $L_1\cap L_2$ is $\reg$-dissectible. More generally,
let $k$ be any number $\geq2$ and let $L_1,L_2,\ldots,L_k$ be $k$ semi-linear languages. If $\left|\bigcap_{i=1}^{k}L_i\right|=\infty$ and $\bigcap_{i=1}^{k}\Psi(L_i)\subseteq \Psi(\bigcap_{i=1}^{k}L_i)$ hold, then $\bigcap_{i=1}^{k}L_i$ is $\reg$-dissectible.
\end{claim}

\begin{proof}
Since $\Psi(L_1\cap L_2)\subseteq \Psi(L_1)\cap\Psi(L_2)$ always holds, our assumption actually  means  $\Psi(L_1\cap L_2) = \Psi(L_1)\cap\Psi(L_2)$. Since the set of all semi-linear sets is closed under Boolean operations (as well as projections) \cite{GS66}, we conclude that $L_1\cap L_2$ belongs to $\semilin$. Lemma \ref{SEMILIN-dissect} implies that $L_1\cap L_2\in\regdissect$.
The above proof can be easily extended to the case of the intersection $\bigcap_{i=1}^{k}\Psi(L_i)$ of $k$ commutative images.
\end{proof}

Now, let $L'=L[w_1,w_2,\ldots,w_m]$ and take any $k$ subsets  $L_1,L_2,\ldots,L_k\in\bcfl$ of $L'$. As noted earlier, it follows that $L_1,L_2,\ldots,L_k\in\semilin$. Here, we assume that $L=\bigcap_{i=1}^{k}L_i$ is an infinite set.
By Claim \ref{semilin-two-intersection},  we only need to prove that
$\bigcap_{i=1}^{k}\tilde{\Psi}(L_i) \subseteq \tilde{\Psi}(\bigcap_{i=1}^{k}L_i)$.  Firstly,  choose  any point $v\in \bigcap_{i=1}^{k}\tilde{\Psi}(L_i)$ and fix $i\in[k]$ arbitrarily.
Since the inverse image  $\tilde{\Psi}^{-1}(v)=\{w\in L' \mid v\in \tilde{\Psi}(w)\}$  must be a singleton, there exists a {\em unique} string $w\in L'$ for which $\tilde{\Psi}^{-1}(v)=\{w\}$. {}From $v\in \tilde{\Psi}(L_i)$, we obtain the membership $w\in L_i$. Moreover, since $i$ is arbitrary, we conclude that $w$ is in $\bigcap_{i=1}^{k} L_i$. It therefore follows that $v\in \tilde{\Psi}(w)\subseteq \tilde{\Psi}(\bigcap_{i=1}^{k}L_i)$.
In conclusion, $L$ is  $\reg$-dissectible.
\end{proofof}

Without the condition $\Psi(L_1)\cap \Psi(L_2)\subseteq \Psi(L_1\cap L_2)$ of Claim \ref{semilin-two-intersection}, nevertheless, it is impossible to prove the intersection of two semi-linear languages to be $\reg$-dissectible since $\semilin(2)\nsubseteq\regdissect$.

Next, we will show the $\reg$-dissectibility of the {\em Boolean hierarchy over BCFL}, where the Boolean hierarchy over $\bcfl$ is defined as follows: $\bcfl_{1} = \bcfl$, $\bcfl_{2k} = \bcfl_{2k-1}\wedge \co\bcfl$, and $\bcfl_{2k+1} = \bcfl_{2k}\vee \bcfl$ for every number  $k\in\nat^{+}$. Finally, we set $\bcfl_{\mathrm{BH}} = \bigcup_{k\geq1}\bcfl_{k}$.

\begin{theorem}\label{Boolen-dissect}
The Boolean hierarchy $\bcflbh$ is $\reg$-dissectible.
\end{theorem}

\begin{proof}
Since $\bcfl_{2k-1}\subseteq \bcfl_{2k}$ holds for every number $k\in\nat^{+}$, it is sufficient to prove that $\bcfl_{2k}$ is $\reg$-dissectible for all indices $k\in\nat$.  We will show this claim by induction on $k$.
For the basis case of $\bcfl_{2}$ ($= \bcfl-\bcfl$), let $L_1$ and $L_2$ be languages over alphabet $\Sigma$ in $\bcfl$ and concentrate on the difference $L_1-L_2$. First, we intend to prove Claim \ref{difference-SEMILIN}. In the claim,  the notation $\|v\|_{1}$ for any vector $v$ in a Euclidean space denotes the {\em $\ell_1$-norm} of $v$; that is, $\|v\|_{1} = \sum_{i}|v_i|$ if $v=(v_i)_{i}$.

\begin{claim}\label{difference-SEMILIN}
Let $L_1$ and $L_2$ be any two infinite semi-linear languages satisfying  $\Psi(L_1)\not\subseteq_{ae}\Psi(L_2)$. If $\Psi(L_1)-\Psi(L_2)\subseteq \Psi(L_1-L_2)$ holds, then the difference $L_1-L_2$ is $\reg$-dissectible.
\end{claim}

\begin{proof}
Since $\Psi(L_1)$ and $\Psi(L_2)$ are both semi-linear, the difference $\Psi(L_1)-\Psi(L_2)$ is   semi-linear as well \cite{GS66}. By our assumption follows the equality $\Psi(L_1)-\Psi(L_2) = \Psi(L_1-L_2)$. There exists a series of critical matrices that characterizes $\Psi(L_1-L_2)$.
Here, we want to fix one of them, say, $T=(v_j)_{1\leq j\leq m}$, where each $v_j$ is a column vector. For simplicity, we assume that $v_1\neq O$ and, moreover,  the second entry of $v_1$ is non-zero.  Given each index $i\in\{0,1\}$, let us consider a set $A_i = \{w\in\Sigma^*\mid \exists z_1\in\nat\,[\,(1,2z_1+i,0,\ldots,0)T = \Psi(w)\,]\}$. Since $\Psi(A_0\cup A_1)\subseteq \Psi(L_1-L_2)$, we conclude that $A_0\cup A_1\subseteq L_1-L_2$.
It is clear that $A_i$ is infinite and the language $C_i = \{w\in\Sigma^*\mid |w|= \| (1,2z_1+i,0,\ldots,0)T  \|_{1}\}$ is also infinite because of $A_i\subseteq C_i$. In addition, $C_i$ is regular because every string $w$ in $C_i$ satisfies $|w|=\|v_0\|_{1}+(2z_1+i)\|v_1\|_{1}$ and it is easy to determine whether or not this is true for any given string $w$ by running an appropriate dfa.  Since $C_0\cap C_1=\setempty$ and $A_i\subseteq C_i\cap (L_1-L_2)$ for each index $i\in\{0,1\}$,  $C_i$ must dissect $L_1-L_2$. Hence, $L_1-L_2$ is $\reg$-dissectible.
\end{proof}

Now, we claim that  $\tilde{\Psi}(L_1)-\tilde{\Psi}(L_2) \subseteq \tilde{\Psi}(L_1-L_2)$  for two arbitrary  languages $L_1$ and $L_2$ in $\bcfl$.  To prove this claim, take any point $v\in\tilde{\Psi}(L_1) - \tilde{\Psi}(L_2)$. Since $v\in\tilde{\Psi}(L_1)$, there exists a string $w\in L_1$ for which $v\in\tilde{\Psi}(w)$. Note that $w\not\in L_2$ because, otherwise, we obtain $v\in\tilde{\Psi}(w)\subseteq \tilde{\Psi}(L_2)$, a contradiction. Since $w\in L_1-L_2$, it follows that $v\in\tilde{\Psi}(w)\subseteq \tilde{\Psi}(L_1-L_2)$.
Using a similarity between $\Psi(w)$ and $\tilde{\Psi}(w)$ as in the proof of Theorem \ref{BCFL(k)-dissect}, we can apply Claim \ref{difference-SEMILIN} and then obtain the $\reg$-dissectibility of $L_1-L_2$.

The remaining task is to deal with the induction case of $\bcfl_{2k}$ for any number $k\geq2$. For this purpose, we will present a simple fact on the even levels of the Boolean hierarchy over $\bcfl$.

\begin{claim}\label{BCFL-character}
For every number $k\geq2$, $\bcfl_{2k} = \bcfl_{2k-2}\vee \bcfl_{2}$.
\end{claim}

\begin{proof}
Here, we want to prove that (*) for every number $k\geq2$,
$\bcfl_{2k-2}\wedge \co\bcfl = \bcfl_{2k-2}$.
Write  $\FF$ for $\bcfl_{2k-2}\wedge \co\bcfl$ for simplicity.  Since $\bcfl_{2k-2}=\bcfl_{2k-3}\wedge \co\bcfl$ holds by the definition, $\FF$ equals $\bcfl_{2k-3}\wedge (\co\bcfl\wedge \co\bcfl)$, which is actually  $\bcfl_{2k-3}\wedge \co(\bcfl\vee \bcfl)$. Since $\bcfl$ is closed under union (\ie $\bcfl\vee \bcfl = \bcfl$),  it follows that $\FF = \bcfl_{2k-3}\wedge \co\bcfl$. By the definition again, the right-hand side of this equation coincides with $\bcfl_{2k-2}$. Therefore, Statement (*) holds.

Recall that $\bcfl_{2k}$ equals $\bcfl_{2k-1}\wedge \co\bcfl$, which also coincides with  $(\bcfl_{2k-2}\vee \bcfl)\wedge \co\bcfl$. By DeMorgan's law, it holds that $\bcfl_{2k} = (\bcfl_{2k-2}\wedge \co\bcfl)\vee (\bcfl\wedge \co\bcfl)$. Statement (*) then leads to $\bcfl_{2k} = \bcfl_{2k-2}\vee \bcfl_{2}$, as requested.
\end{proof}

Notice that the induction hypothesis ensures the $\reg$-dissectibility of  $\bcfl_{2k-2}$. Since $\bcfl_{2}$ has been already proven to be  $\reg$-dissectible, $\bcfl_{2k-2}\vee \bcfl_{2}$ must be  $\reg$-dissectible by the closure property of $\regdissect$ discussed in Section \ref{sec:how-to-dissect}. By Claim \ref{BCFL-character}, this family is exactly $\bcfl_{2k}$. This completes the proof of Theorem \ref{Boolen-dissect}
\end{proof}

%%%%%%%%%%%%%%%%%%%%%%%%%
%%%%%%%%%%%%%%%%%%%%%%%%%
\section{Separation with Infinite Margins}\label{sec:application}

In this final section, we will seek a meaningful application of our previous results regarding  the
$\reg$-dissectibility of certain bounded languages.
To describe this application, we need to introduce extra terminology.
Given two infinite sets $A$ and $B$, we say that $A$ {\em covers $B$ with an infinite margin} ($A$ {\em i-covers} $B$, or $A$ is an {\em i-cover} of $B$, in short)
if  both $B\subseteq A$ and   $A\neq_{ae}B$ hold.
When $A$ i-covers $B$, we briefly write $\mathrm{i}(B,A)$ and call it an {\em i-covering pair}.
A language $C$ is said to {\em separate $\mathrm{i}(B,A)$ with infinite margins} (or {\em i-separate} $\mathrm{i}(B,A)$, in short)  if (i) $B\subseteq C \subseteq A$,
(ii) $A\neq_{ae}C$, and (iii) $B\neq_{ae}C$.
For convenience, we use the notation $\mathrm{i}(\BB,\AAA)$ for two language families $\AAA$ and $\BB$ to denote the set of all i-covering pairs $\mathrm{i}(B,A)$ satisfying  $A\in\AAA$ and  $B\in\BB$.
Another language family $\CC$ is said to {\em i-separate} $\mathrm{i}(\BB,\AAA)$ if, for every pair $\mathrm{i}(B,A)$ in $\mathrm{i}(\BB,\AAA)$, there exists a set in $\CC$ that i-separates $\mathrm{i}(B,A)$.

The following is a key lemma that bridges between the $\reg$-dissectibility and the i-separation.

\begin{lemma}\label{dissect-imply-separate}
Let $\AAA$ and $\BB$ be any two language families and assume that $\AAA-\BB$ is $\reg$-dissectible.   It then holds that, for any $A\in\AAA$ and any $B\in\BB$, if $A$ i-covers $B$, then there exists a language in $\EE$ that i-separates $\mathrm{i}(B,A)$, where $\EE$ expresses the set $\{B\cup(A\cap C)\mid A\in\AAA,B\in\BB,C\in\reg\}$. In other words, $\EE$ i-separates $\mathrm{i}(\BB,\AAA)$.
\end{lemma}

\begin{yproof}
Let $A\in\AAA$ and $B\in\BB$ be two infinite languages. Let $D=A-B$ and assume that $D$ is infinite. Our assumption guarantees the existence of a regular language $C$ for which  $C$ dissects $D$. For convenience, we set $E = B \cup (A \cap C)$.  Since $C$ dissects $D$, it follows that $|(A\cap C)-B|=\infty$ and $|(A \cap\overline{C})-B|=\infty$. These conditions imply that $B\subseteq  E\subseteq A$ and $|A-E|=|E-B|=\infty$. Thus, $E$ i-separates $\mathrm{i}(B,A)$.  Since $C$ is regular, $E$ clearly belongs to the language family $\EE$.
\end{yproof}

Concerning bounded context-free languages, we can show the following i-separation result.

\begin{theorem}\label{BCFL-separate}
For any index $k\in\nat^{+}$, $\bcfl_{k}$ i-separates $\mathrm{i}(\bcfl_{k},\bcfl_{k})$.  Thus, $\bcflbh$ i-separates $\mathrm{i}(\bcflbh,\bcflbh)$.
\end{theorem}

\begin{proof}
Hereafter, we intend to show that $\bcfl_{k}-\bcfl_{k}$ is $\reg$-dissectible because an application of Lemma \ref{dissect-imply-separate} immediately leads to the theorem. For our purpose,  it suffices to  prove that $\bcfl_{k}-\bcfl_{k}$ is included in $\bcflbh$, because $\bcflbh$ is $\reg$-dissectible by Theorem \ref{BCFL(k)-dissect}. More strongly, we  will  demonstrate that,
for any two indices $i,j\geq1$, $\bcfl_{i}-\bcfl_{j}\subseteq \bcflbh$.

Given an index pair $(i,j)\in\nat^{+}\times\nat^{+}$, let $\FF_{i,j} = \bcfl_{i}-\bcfl_{j} = \bcfl_{i}\wedge \co\bcfl_{j}$ and $\GG_{i,j} = \bcfl_{i}\wedge \bcfl_{j}$ for simplicity.
We will show that $\FF_{i,j}\subseteq \bcflbh$ by induction on $(i,j)$.
For the basis case $(1,1)$, since $\FF_{1,1} = \bcfl_{2}$ holds, clearly $\FF_{1,1}$ is a subset of $\bcflbh$.
For the second case $(2,1)$, we first note that $\bcfl_{4} = (\bcfl_{2}\wedge \co\bcfl_{2})\vee (\bcfl_{2}\wedge \bcfl_{2}) = \FF_{2,1}\vee \GG_{2,2}$. We thus obtain  $\FF_{2,1}\subseteq \bcfl_{4}$ as well as $\GG_{2,2}\subseteq \bcfl_{4}$.
For the induction case $(i,j)$, it is enough to consider the case where $i=2k$ and $j=2m+1$.
Similar to Claim \ref{BCFL-character}, we can prove the next useful relation.

\begin{claim}\label{coBCFL-character}
$\co\bcfl_{2k+1} = \bcfl_{2k-1}\vee \bcfl_{2}$.
\end{claim}

By Claims \ref{BCFL-character} and \ref{coBCFL-character}, $\FF_{2k,2m+1}$ equals $(\bcfl_{2k-2}\vee \bcfl_{2})\wedge ( \co\bcfl_{2m-1}\vee \bcfl_{2})$, which  can be
transformed into $\FF_{2k-2,2m-1}\vee \FF_{2,2m-1}\vee \GG_{2k-2,2}\vee \GG_{2,2}$.
By the induction hypothesis, there are two indices $\ell_1,\ell_2$ such that $\FF_{2k-2,2m-1}\subseteq \bcfl_{2\ell_1}$ and $\FF_{2,2m-1}\subseteq \bcfl_{2\ell_2}$.
By applying Claim \ref{BCFL-character} repeatedly, we then obtain $\bcfl_{2\ell_1} = \bigvee_{i=1}^{\ell_1}\bcfl_{2}$ and $\bcfl_{2\ell_2} = \bigvee_{i=1}^{\ell_2}\bcfl_{2}$.
Likewise, we obtain $\bcfl_{2k-2} = \bigvee_{i=1}^{k-1}\bcfl_{2}$. Hence, $\GG_{2k-2,2}$ equals $(\bigvee_{i=1}^{k-1}\bcfl_{2})\wedge \bcfl_{2} = \bigvee_{i=1}^{k-1}\GG_{2,2}$, which is included in $\bigvee_{i=1}^{k-1}\bcfl_{4} = \bcfl_{4(k-1)}$.
This fact implies the containment $\GG_{2k-2,2}\vee \GG_{2,2}\subseteq \bcfl_{4k}$. It thus  follows that $\FF_{2k,2m+1}\subseteq \bcfl_{2\ell_1}\vee \bcfl_{2\ell_2} \vee \bcfl_{4k} = \bigvee_{i=1}^{\ell_1+\ell_2+2k}\bcfl_{2}$. As discussed before, this is equivalent to $\bcfl_{2(\ell_1+\ell_2+2k)}$, which is obviously included in $\bcflbh$.
Therefore, we conclude that $\FF_{2k,2m+1}\subseteq \bcflbh$.
\end{proof}

%%%%%%%%%%%%%%%%%%%%%%%%%%%%%%
%%%%%%%%%%%%%%%%%%%%%%%%%%%%%%
\section{Future Challenges}

We have initiated a fundamental study on the dissecting power of regular languages and an  application of the $\reg$-dissectibility to the i-separation.
Throughout our initial study,  a number of open questions have arisen for future research. An important open question concerns the $\reg$-dissectibility of $\co\cfl$ and, more widely,      $\cfl_{k}$ and $\cfl(k)$, which are respectively $\cfl$-analogues of  $\bcfl_{k}$ and $\bcfl(k)$,  for every level $k\geq2$. Slightly apart from $\cfl$,  two other language families $\onecequallin$ and $\oneplin$, introduced in  \cite{TYL10}, are, at this moment, unknown to be $\reg$-dissectible. Much anticipated is a development of a coherent theory of a more general notion of $\CC$-dissectibility.
Concerning the i-separation of $\mathrm{i}(\cfl,\cfl)$, on the contrary, a key question of whether $\cfl$ i-separate $\mathrm{i}(\cfl,\cfl)$ still awaits its answer. Lately, we have learned that Bucher \cite{Buc80} had raised essentially the same question back in 1980.

%%%%%%%%%%%%%%%%%%

\bigskip\bigskip
\paragraph{Acknowledgments}
The first author is grateful to Jeffrey Shallit for drawing his attention to \cite{DSY02} whose core concept has helped formulate an initial notion of ``dissectibility'' and to Jacobo Tor\'{a}n and a reviewer for pointing to \cite{Buc80} and providing its hard copy in the last moment.

%%%%%%%%%%%%%%%%%%%%%%%%%%
%%%%%%%%%%%%%%%%%%%%%%%%%%
\bibliographystyle{alpha}

%%%%%%%%%%%%%%%%%%%%%%%%%%%%%%%%%%%%%%%%%%%%%%%%%%%
%%%%%%%%%%%%%%%%%%%%%%%%%%%%%%%%%%%%%%%%%%%%%%%%%%%
\end{document}